\pdfoutput=1
\RequirePackage{ifpdf}
\ifpdf 
\documentclass[pdftex]{sigma}
\else
\documentclass{sigma}
\fi

\usepackage{mathtools}

\DeclareMathOperator{\cyc}{cyc}

\DeclareMathOperator{\aut}{aut}

\DeclareMathOperator{\Li}{Li}

\newtheorem{Theorem}{Theorem}[section]
 { \theoremstyle{definition}
\newtheorem{Remark}[Theorem]{Remark} }

\numberwithin{equation}{section}

\DeclarePairedDelimiter{\abs}{|}{|}

\def\ss{\subset}
\def\deq{\coloneqq}

\def\ss{\subset}

\def\&{&{\hskip -20pt}}

\def\JJ{\mathcal{J}}

\def\Cb{\mathbf{C}}
\def\Ib{\mathbf{I}}

\def\Nb{\mathbf{N}}
\def\Pb{\mathbf{P}}
\def\Rb{\mathbf{R}}

\def\Zb{\mathbf{Z}}

\begin{document}

\newcommand{\arXivNumber}{1504.07512}

\allowdisplaybreaks

\renewcommand{\thefootnote}{$\star$}

\renewcommand{\PaperNumber}{097}

\FirstPageHeading

\ShortArticleName{Multispecies Weighted Hurwitz Numbers}

\ArticleName{Multispecies Weighted Hurwitz Numbers\footnote{This paper is a~contribution to the Special Issue on Exact Solvability and Symmetry Avatars
in honour of Luc Vinet.
The full collection is available at
\href{http://www.emis.de/journals/SIGMA/ESSA2014.html}{http://www.emis.de/journals/SIGMA/ESSA2014.html}}}

\Author{J.~{HARNAD}~$^{\dag\ddag}$}

\AuthorNameForHeading{J.~Harnad}

\Address{$^\dag$~Centre de recherches math\'ematiques, Universit\'e de Montr\'eal,\\
\hphantom{$^\dag$}~C.P.~6128, succ.\ Centre-ville, Montr\'eal (QC) H3C 3J7, Canada}
\EmailD{\href{mailto:harnad@crm.umontreal.ca}{harnad@crm.umontreal.ca}}

\Address{$^\ddag$~Department of Mathematics and Statistics, Concordia University,\\
\hphantom{$^\ddag$}~7141 Sherbrooke~W., Montr\'eal (QC) H4B 1R6, Canada}

\ArticleDates{Received March 31, 2015, in f\/inal form November 16, 2015; Published online December 02, 2015}

\Abstract{The construction of hypergeometric $2D$ Toda $\tau$-functions as generating functions for weighted Hurwitz numbers is extended  to multispecies families.  Both the enumerative geometrical signif\/icance of multispecies weighted Hurwitz numbers, as weighted enumerations of  branched coverings of the Riemann sphere, and their combinatorial  signif\/icance in terms of weighted paths in the Cayley graph of~$S_n$ are derived. The particular case of multispecies quantum weighted Hurwitz numbers is studied in detail.}


\Classification{05A15; 14H30; 33C70; 57M12}

\renewcommand{\thefootnote}{\arabic{footnote}}
\setcounter{footnote}{0}

\section{Introduction}

  In \cite{GH1, GH2} a simple method was developed for constructing parametric families  of KP and
$2D$ Toda $\tau$-functions~\cite{Ta, Takeb, UTa} of  hypergeometric type~\cite{GR,OrSc} that  serve
as generating functions for the weighted enumeration of $n$-sheeted  branched coverings of the Riemann sphere.
These are characterized by the fact that their expansions in the basis of products of Schur functions have only diagonal  coef\/f\/icients, and these are of  special {\em content product} form. When expanded instead  in the  basis of products of power sum symmetric functions, the coef\/f\/icients turn out to be weighted  sums of Hurwitz numbers, with weighting dependent  generally  on an inf\/inite sequence of parameters ${\bf c}= (c_1, c_2, \dots)$  determined by an associated {\em weight generating function}~$G(z)$. Such weighted sums may be interpreted equivalently  as weighted enumeration  of paths in the Cayley graph of the symmetric group $S_n$  generated by transpositions.

The special choice $G(z)=\exp(z)$ gives rise to the generating functions for simple and double Hurwitz numbers
introduced  by  Pandharipande~\cite{Pa} and Okounkov~\cite{Ok}, in which all branch points other than  a single one  or a specif\/ied pair
have simple branching prof\/iles. Further insight into the signif\/icance of these hypergeometric $\tau$-functions as generating
functions in terms of recurrence relations for  triangulations was developed in~\cite{GJ}. More generally, the ``topological recursion'' program
was successfully applied to the study of Hurwitz numbers in \cite{BM,EO3,EO1,EO2}.

It has been known since the pioneering works of Hurwitz~\cite{Frob1, Frob2}, Frobenius~\cite{Hu1, Hu2} and Schur~\cite{Sch} that these numbers  may be reinterpreted combinatorially using the monodromy representation of the fundamental group of the punctured sphere with  values in $S_n$ determined by lifting closed paths to the covering surface. From this viewpoint, they enumerate factorizations of the identity element into products of elements whose
conjugacy classes correspond to the ramif\/ications prof\/iles, and hence give uniformly weighted enumeration of paths in the Cayley graph from one conjugacy class to another with a given number of steps.

Another choice of weight generating function that was studied in~\cite{GH2}  is the quantum di\-logarithm function~\cite{FK}. This amounts to equating the parameters~$c_i$ to powers of  a single quantum deformation parameter~$q$, and gives rise to four special versions  of~quantum weighted Hurwitz numbers, whose distribution functions were linked to those of a Bosonic gas with  linear energy spectrum in~\cite{GH2}. Using a suitably extended class of weight generating functions, this was further generalized in~\cite{H1, H2}  to include both the inf\/inite family of  classical weighting parameters~${\bf c}$ and the pair~$(q,t)$ of quantum deformation parameters characterizing  Macdonald polynomials.

In the present work, the notion of weighted Hurwitz numbers is extended to weighted enu\-me\-ra\-tions
of multispecies coverings  involving arbitrary choices for the corresponding weighting para\-meters.
The generating function depends on $l+m$ expansion parameters $(w_1, \dots, w_l; z_1, \dots , z_m)$,
corresponding to two classes and $l+m$ subspecies (or ``colours'') of branch points
with rami\-f\/i\-ca\-tion prof\/iles types $\big\{\mu^{(\beta)}\big\}_{\beta =1 , \dots , l}$ and  $\big\{\nu^{(\beta)}\big\}_{\beta =1 , \dots , m}$.
  The special case of signed multispecies enumeration in the uniformly weighted case
 was  studied in~\cite{HOr}. Its combinatorial  signif\/icance was explained  in terms of  enumeration of
 paths in the Cayley graph that are subdivided into strictly or weakly  monotonically increasing subsequences of  transpositions having given lengths.  In the single species case~\cite{GH2}, this was shown to be equivalent  to (signed) enumeration of branched covers of the Riemann sphere with  the
 ``coloured'' branch points constrained to have f\/ixed total ramif\/ication index within each  class.
  Another special case detailed here consists of ``multispecies quantum weighted
 Hurwitz numbers'',  in which the weighting parameters consist of  powers of a~sequence of
 auxiliary  quantum deformation parameters ${\bf q} = (q_1, \dots, q_l)$, ${\bf p} = (p_1, \dots, p_m)$.

 In  Sections~\ref{hurwitz_numbers}--\ref{hypergeometric_tau_generators} the basic notions regarding Hurwitz numbers will be  recalled, together with the construction of weighted Hurwitz numbers  using inf\/inite parameter families of weight generating  functions~$G(z)$, as developed in~\cite{GH1, GH2, HOr}.  In Section~\ref{generating_multispecies_weighted_hurwitz}, the single species case will be extended to
 multispecies by introducing the idea of ``coloured'' branch points, of two classes.
Weight generating functions depending on a multiparametric set of expansion parameters multiplicatively
provide multiparametric families of  $2D$ Toda $\tau$-functions of hypergeometric type that are generating functions for  such multispecies weighted Hurwitz numbers. As in the single  species case, these may be viewed both geometrically and combinatorially, in terms of weighted coverings and paths in the Cayley graph.  In Section~\ref{multispecies_quantum_hurwitz}  this is restricted to the special cases of quantum weightings introduced in \cite{GH2} to produce generating functions for multispecies quantum Hurwitz numbers. These are interpreted, as in the single  species case, both geometrically and
 combinatorially\footnote{A dif\/ferent multiparametric family of generating functions for weighted Hurwitz numbers over $\Rb\Pb^2$, consisting  of BKP generating functions of  hypergeometric type, was considered in~\cite{NOr1, NOr2}.}.

\subsection{Hurwitz numbers}
\label{hurwitz_numbers}

 For a set of $k\in \Nb^+$ partitons $\big(\mu^{(1)}, \dots, \mu^{(k)}\big)$ of $n\in\Nb^+$, let  $H\big(\mu^{(1)}, \dots, \mu^{(k)}\big)$
 denote the number of  $n$-sheeted branched coverings  of the Riemann sphere (not necessarily connected), with~$k$ branch points,  whose ramif\/ication prof\/iles are given by the  partitions, weighted  by the inverse of the order of   the automorphism group. These are  the geometrically def\/ined  {\em  Hurwitz numbers}  as originally studied by Hurwitz \cite{Hu1, Hu2, LZ}. The genus~$g$ of the covering surface is determined by the Riemann Hurwitz formula for the Euler characteristic
 \begin{gather}
 2- 2g = 2n - \sum_{i=1}^k \ell^*\big(\mu^{(i)}\big),
 \end{gather}
 where
 \begin{gather}
 \ell^*(\mu) := |\mu| -\ell(\mu)
 \end{gather}
 is the {\em colength} of the partition~$\mu$, i.e.,   the complement of the length~$\ell(\mu)$  with respect to its weight~$|\mu|$,
 or  the degree of degeneracy of the branched cover over a~point with ramif\/ication prof\/ile type~$\mu$.

The Frobenius--Schur formula \cite{Frob1, Frob2, LZ, Sch} expresses these as sums over  irreducible
 $S_n $ cha\-rac\-ters
  \begin{gather}
  H\big(\mu^{(1)}, \dots, \mu^{(k)}\big) =
   \sum_{\lambda, |\lambda|=n} h_\lambda^{k-2} \prod_{i=1}^k {\chi_\lambda(\mu^{(i)}) \over z_{\mu^{(i)}}},
   \label{Frobenius_Schur_Hurwitz}
  \end{gather}
where $\chi_\lambda(\mu)$ is the character of the  irreducible representation of symmetry  type $\lambda$, evaluated on the
conjugacy class $\cyc(\mu)$ consisting of elements with cycle lengths equal to the parts of the partition~$\mu$,
\begin{gather}
h_\lambda = \det\left({1\over{( \lambda_i -i +j)!}}\right)^{-1}
\end{gather}
is the product  of the hook lengths in the Young diagram of  partition $\lambda$ and
\begin{gather}
z_\mu = \prod_{i\in \Nb} i^{m_i(\mu)} (m_i(\mu))!
\end{gather}
is the order of the stabilizer under conjugation of any element of the conjugacy class~$\cyc(\mu)$,
with~$m_i(\mu)$   the number of parts of the  partition~$\mu$ equal to~$i$,

There is an alternative interpretation of $H\big(\mu^{(1)}, \dots, \mu^{(k)}\big)$ that is purely combinatorial;
it equals~${1\over n!}$ times the number of ways in which the identity element $\Ib \in S_n$ may
be expressed as a~pro\-duct of $k$ elements belonging to the conjugacy classes of
cycle type  $\big\{\cyc\big(\mu^{(i)}\big)\big\}_{i=1, \dots , k}$
\begin{gather}
\Ib = g_1 g_2 \cdots g_k, \qquad \text{where} \quad g_i\in \cyc\big(\mu^{(i)}\big).
\end{gather}
The two  are related by noting that each such factorization may be understood as def\/ining the
image in~$S_n$ of the identity element
in the fundamental group of the punctured sphere with the branch points removed under the monodromy map
obtained by lifting closed loops to the covering surface.

\subsection{Weighted geometrical Hurwitz numbers}
\label{weighted_hurwitz_geometrical}

As def\/ined in \cite{GH2}, given a weight generating function  $G(z)$ expressible as an inf\/inite product
\begin{gather}
G(z) = \prod_{i=1}^\infty (1+ c_i z), \qquad {\bf c} =(c_1, c_2, \dots ),
\end{gather}
the weight $W_G\big(\mu^{(1)}, \dots, \mu^{(k)}\big)$ assigned to a conf\/iguration of $k+2$ branch points
 with ramif\/ication prof\/iles $\big\{\mu^{(1)}, \dots, \mu^{(k)}, \mu, \nu\big\}$ is solely determined by the colengths
 $\big\{\ell^*\big(\mu^{(1)}\big),  \dots, \ell^*\big(\mu^{(k)}\big)\big\}$, and is given by evaluation of the monomial
 sum symmetric functions at the parameter values
 \begin{gather}
W_G\big(\mu^{(1)}, \dots, \mu^{(k)}\big) := m_\lambda ({\bf c}) =
\frac{1}{\abs{\aut(\lambda)}} \sum_{\sigma\in S_k} \sum_{1 \le i_1 < \cdots < i_k}
 c_{i_{\sigma(1)}}^{\ell^*(\mu^{(1)})} \cdots c_{i_{\sigma(k)}}^{\ell^*(\mu^{(k)})}.
\label{WG_weight}
\end{gather}
Here $\lambda$ is the partition of  length  $\ell(\lambda)=k$ whose parts are $\big\{\ell^*\big(\mu^{(i)}\big)\big\}_{i=1, \dots, k}$,
and $\abs{\aut(\lambda)} $ is the order of the automorphism group of $\lambda${\samepage
\begin{gather}
|\aut(\lambda)| := \prod_{i=1}^{\ell(\lambda} (m_i(\lambda))!,
\end{gather}
where $m_i(\lambda)$ is the number of parts of $\lambda$ equal to~$i$.}

The geometrically def\/ined double  weighted   Hurwitz numbers $H^d_G(\mu, \nu)$
give a  weighted enumeration of $n$-sheeted branched covers of the Riemann sphere that contain
a pair of f\/ixed branch points, say  at $(0, \infty)$, with ramif\/ication prof\/ile types given by the pair of partitions~$(\mu, \nu)$
and a further set of $k$ branch points with ramif\/ication prof\/iles $ \big(\mu^{(1)}, \dots, \mu^{(k)}\big)$. They are def\/ined
by the weighted sums
\begin{gather}
H^d_G(\mu, \nu) \deq \sum_{k=0}^\infty \sideset{}{'}\sum_{\substack{\mu^{(1)}, \dots, \mu^{(k)} \\ \sum\limits_{i=1}^k \ell^*(\mu^{(i)})= d}}
W_G\big(\mu^{(1)}, \dots, \mu^{(k)}\big) H\big(\mu^{(1)}, \dots, \mu^{(k)}, \mu, \nu\big),
\end{gather}
over all $k$-tuples of nontrivial ramif\/ication prof\/iles satisfying the condition
\begin{gather}
d = \sum_{i=1}^k \ell^*\big(\mu^{(i)}\big) =|\lambda|,
\end{gather}
with weight $W_G\big(\mu^{(1)}, \dots, \mu^{(k)}\big) $  given by~(\ref{WG_weight}).

The Riemann--Hurwitz formula for the genus $g$ of the covering surface is then
\begin{gather}
2-2g = \ell(\mu) + \ell(\nu) -d.
\end{gather}

 An alternative is to use the dual weight generating function
 \begin{gather}
 \tilde{G}(z)= \prod_{i=1}^\infty (1- \tilde{c}_i z)^{-1},  \qquad \tilde{\bf c} =(\tilde{c}_1, \tilde{c}_2, \dots ),
 \end{gather}
 for which the geometrical weight $W_{\tilde{G}}\big(\mu^{(1)}, \dots, \mu^{(k)}\big)$
 is given by the ``forgotten'' symmetric function $f_\lambda(\tilde{\bf c})$, $\tilde{\bf c}=(\tilde{c}_1, \tilde{c}_2, \dots )$,
\begin{gather}
W_{\tilde{G}}\big(\mu^{(1)}, \dots, \mu^{(k)}\big)  :=
f_\lambda (\tilde {\bf c})=
\frac{(-1)^{\ell^*(\lambda)}}{\abs{\aut(\lambda)}} \sum_{\sigma\in S_k} \sum_{1 \le i_1 \le \cdots \le i_k}
\tilde c_{i_{\sigma(1)}}^{\ell^*(\mu^{(1)})}  \cdots \tilde c_{i_{\sigma(k)}}^{\ell^*(\mu^{(k)})},
   \label{dual_WG_weight}
 \end{gather}
 where the partition $\lambda$ is again def\/ined as above,  with parts consisting of the colengths\linebreak
 $\{\ell^*(\mu^{(i)})\}_{i=1, \dots, k}$.  The dually weighted geometrical Hurwitz numbers are similarly def\/ined by the weighted sum
 \begin{gather}
H^d_{\tilde{G}}(\mu, \nu) \deq \sum_{k=0}^\infty \sideset{}{'}\sum_{\substack{\mu^{(1)}, \dots, \mu^{(k)} \\ \sum\limits_{i=1}^k \ell^*(\mu^{(i)})= d}}
W_{\tilde{G}}\big(\mu^{(1)}, \dots, \mu^{(k)}\big) H\big(\mu^{(1)}, \dots, \mu^{(k)}, \mu, \nu\big).
\end{gather}

\subsection{Weighted combinatorial Hurwitz numbers}
\label{weighted_hurwitz_combinatorial}

Following \cite{GH1, GH2}, we may alternatively def\/ine a combinatorial Hurwitz number $F^d_G(\mu, \nu)$
that gives the weighted enumeration of $d$-step paths in the Cayley graph of the symmetric
group $S_n$ generated by transpositions $(a,b)$, $b>a$, starting at an element $h\in \cyc(\mu)$ in the
conjugacy class $\cyc(\mu)$ consisting of elements with  cycle lengths equal to the parts of
$\mu$ and ending in the conjugacy class $\cyc(\nu)$
\begin{gather}
(a_d b_d) \cdots (a_1 b_1)h \in \cyc(\nu).
\end{gather}
Every such path has a {\em signature} $\lambda$, which is def\/ined to be the partition of weight~$d$,
whose parts are, in weakly decreasing order, the number of times any given second element~$b_i$,
$i=1, \dots , \ell(\lambda)$ is repeated. In the case of the weight generating function~$G(z)$, we assign to
any  path with signature $\lambda$ a combinatorial weight  equal to the product~$e_\lambda({\bf c})$
of the elementary  symmetric functions~\cite{Mac}, evaluated at the parameters $(c_a, c_2, \dots)$
\begin{gather}
e_\lambda({\bf c}) = \prod_{i=1}^{\ell(\lambda)} e_{\lambda_i}({\bf c}).
\end{gather}
In the case of the dual generating functions $\tilde{G}(z)$, we assign a combinatorial weight
equal to the product  $h_\lambda({\bf c})$ of the complete symmetric functions~\cite{Mac},
evaluated at the parameters $(\tilde{c}_1, \tilde{c}_2, \dots)$
\begin{gather}
h_\lambda(\tilde{\bf c}) = \prod_{i=1}^{\ell(\lambda)} h_{\lambda_i}(\tilde{\bf c}) .
\end{gather}

Let $m^\lambda_{\mu \nu}$ be the number of $d=|\lambda|$ step paths of signature $\lambda$
starting at $h \in \cyc(\mu)$ and ending in the conjugacy class $\cyc(\nu)$.
Then the combinatorial weighted Hurwitz numbers~$F^d_G(\mu, \nu)$, $F^d_{\tilde{G}}(\mu, \nu)$ are
 def\/ined to be the weighted sums
\begin{gather}
F^d_G(\mu, \nu)  := {1\over n!} \sum_\lambda e_\lambda({\bf c}) m^\lambda_{\mu \nu},
\label{Fd_G_mu_nu}
 \\
F^d_{\tilde{G}}(\mu, \nu)   := {1\over n!} \sum_\lambda h_\lambda(\tilde{\bf c}) m^\lambda_{\mu \nu}.
\label{Fd_tilde_G_mu_nu}
\end{gather}
In \cite{GH2}  it is proved that these two notions of weighted Hurwitz numbers in fact coincide:

\begin{Theorem}[\cite{GH2}]
\begin{gather}
F^d_G(\mu, \nu) = H^d_G(\mu, \nu), \qquad F^d_{\tilde{G}}(\mu, \nu) = H^d_{\tilde{G}}(\mu, \nu).
\label{Fd_G_mu_nu_Hd}
\end{gather}
\end{Theorem}

The main idea behind the proof is  to def\/ine associated elements  $G_n(w, \JJ) \in \Zb(\Cb[S_n])$  and $\tilde{G}_n(z, \JJ)\in \Zb(\Cb[S_n])$
in the centre  $\Zb(\Cb[S_n])$ of the group algebra $\Cb[S_n]$ by
 \begin{gather}
G_n(z, \JJ) :=\prod_{a=1}^n G(z\JJ_a) ,\qquad \tilde{G}(z, \JJ) :=\prod_{a=1}^n \tilde{G}(z\JJ_a),
\end{gather}
where $\JJ:= (\JJ_1, \dots, \JJ_n)$ are the Jucys--Murphy elements \cite{DG, Ju, Mu}
\begin{gather}
\JJ_1:= 0,  \qquad \JJ_b :=\sum_{a=1}^{b-1} (a b), \qquad b=1, \dots , n,
\label{jucys_murphy}
\end{gather}
which generate an abelian subalgebra of the centre $\Zb(\Cb[S_n])$ of the group algebra $\Cb[S_n]$.

The elements $G_n(z, \JJ) $ and  $ \tilde{G}_n(z, \JJ)$  def\/ine  endomorphisms of $\Zb(\Cb[S_n])$ under multiplication,
which are diagonal in the basis  $\{F_\lambda\}$ of $\Zb(\Cb[S_n])$ consisting of the orthogonal idempotents,
 corresponding to irreducible representations, labelled by partitions~$\lambda$ of~$n$
\begin{gather}
G_n(z, \JJ) F_\lambda =r_\lambda^{G(z)}F_\lambda, \qquad \tilde{G}_n(z, \JJ) F_\lambda =r_\lambda^{\tilde{G}(z)}F_\lambda,
 \end{gather}
 with eigenvalues of the following {\em content product} form
 \begin{gather}
r_\lambda^{G(z)}(N) \deq \prod_{(i,j)\in \lambda} G(z(N+ j-i)), \qquad r_\lambda^{\tilde{G}(z)}(N) \deq \prod_{(i,j)\in \lambda} \tilde{G}(z(N+ j-i)).
\end{gather}

 On the other hand, the Cauchy--Littlewood generating function relation \cite{Mac} and its dual show that $G_n(z, \JJ) $ and  $ \tilde{G}_n(z, \JJ)$ may be expanded in terms  of dual bases of the algebra of symmetric functions, evaluated either on the parameters ${\bf c}$ or on the Jucys--Murphy elements $\JJ = (\JJ_1, \dots, \JJ_n)$
 \begin{gather}
 G_n(z, \JJ) = \sum_{\lambda, \, \abs{\lambda}=n} e_\lambda({\bf c}) m_\lambda(\JJ) = \sum_{\lambda,  \, \abs{\lambda}=n} m_\lambda({\bf c}) e_\lambda(\JJ),
 \label{Gn_z_JJ}
 \\
  \tilde{G}_n(z, \JJ) = \sum_{\lambda, \,  \abs{\lambda}=n} h_\lambda({\bf c}) m_\lambda(\JJ) = \sum_{\lambda, \,  \abs{\lambda}=n} f_\lambda({\bf c})
  e_\lambda(\JJ) ,
   \label{G_tilde_n_z_JJ}
 \end{gather}
where $e_\lambda$, $h_\lambda$, $m_\lambda$ and $f_\lambda$ are the elementary, complete, monomial and ``forgotten'' symmetric functions~\cite{Mac},
respectively.  Applying (\ref{Gn_z_JJ}) and (\ref{G_tilde_n_z_JJ})  to the  basis for~$\Zb(\Cb[S_n])$  consisting of the cycle sums
\begin{gather}
C_\mu := \sum_{h\in \cyc(\mu)} h,
\end{gather}
and using the identities
\begin{gather}
e_j(\JJ) = \sum_{\mu, \,  \ell^*(\mu)=j} C_\mu
\qquad \text{and} \qquad
m_\lambda(\JJ) C_\mu = \sum_{\nu, \, \abs{\nu} = \abs{\mu}} m^\lambda_{\mu \nu}C_\nu
\end{gather}
leads to~(\ref{Fd_G_mu_nu_Hd}).

 The bases $\{F_\lambda\}_{\abs{\lambda}=n}$ and  $\{C_\mu\}_{\abs{\mu}=n}$ are related by
\begin{gather}
F_\lambda = h_\lambda \sum_{\mu, \, |\mu| = |\lambda|} \chi_\lambda(\mu) C_\mu,
\end{gather}
where $\chi_{\lambda}(\mu)$ denotes the irreducible character of the irreducible
representation of type  $\lambda$ eva\-luated on the conjugacy class $\cyc(\mu)$.
Under the characteristic map, this   is equivalent to the Frobenius character formula~\cite{Mac}
\begin{gather}
s_\lambda = \sum_{\mu, \, |\mu| = |\lambda|} z_\mu^{-1}\chi_\lambda(\mu)  p_\mu .
\label{Frobenius_character}
\end{gather}

\subsection[Hypergeometric $2D$ Toda $\tau$-functions as generating functions]{Hypergeometric $\boldsymbol{2D}$ Toda $\boldsymbol{\tau}$-functions as generating functions}
\label{hypergeometric_tau_generators}

As shown in \cite{GH2}, for any given generating function of type $G(z)$ or $\tilde{G}(z)$, there is a naturally
associated $2D$ Toda $\tau$-function of hypergeometric type, expressible as a diagonal double Schur function expansion
\begin{gather}
\tau^{G(z)}(N, {\bf t}, {\bf s})  := \sum_{\lambda} r_\lambda^{G(z)}(N) s_\lambda({\bf t}) s_\lambda({\bf s}),
\label{tau_G_double_schur}
\\
\tau^{\tilde{G}(z)}(N, {\bf t}, {\bf s})  := \sum_{\lambda} r_\lambda^{\tilde{G}(z)}(N) s_\lambda({\bf t}) s_\lambda({\bf s}),
\label{tau_G_tilde_double_schur}
\end{gather}
where
\begin{gather}
{\bf t} = (t_1, t_2, \dots), \qquad {\bf s} = (s_1, s_2, \dots)
\end{gather}
are the $2D$ Toda f\/low variables, which may be identif\/ied in this notation
with the power sums
\begin{gather}
t_i = \frac{p_i}{i}, \qquad s_i = \frac{p'_i}{i},  \qquad N \in \Zb
\end{gather}
in two independent sets of variables.
(See \cite{Mac} for notation and further def\/initions involving symmetric functions.)
The coef\/f\/icients have the standard {\em content product} form that characterize such
 $2D$ $\tau$-functions of hypergeometric kind
 \begin{gather}
r_\lambda^{G(z)}(N) \deq r^{G(z)}_0(N) \prod_{(i,j)\in \lambda} G(z(N+ j-i)),
 \end{gather}
where
\begin{gather}
r^{G(z)}_0(N) := \prod_{j=1}^{N-1} G((N-j)z)^j, \qquad r^{G(z)}_0(0)  := 1, \\
 r^{G(z)}_0(-N) := \prod_{j=1}^{N} G((j-N)z)^{-j},
\qquad N\geq 1,
\end{gather}
and identical formulae for $G$ replaced by  $\tilde{G}$.

The other main result of \cite{GH2} is that the resulting $\tau$-functions~(\ref{tau_G_double_schur}),~(\ref{tau_G_tilde_double_schur}),
for $N=0$
\begin{gather}
\tau^{G(z)}({\bf t}, {\bf s}) := \tau^{G(z)}(0, {\bf t}, {\bf s}), \qquad
\tau^{\tilde{G}(z)}( {\bf t}, {\bf s}) :=\tau^{\tilde{G}(z)}(0, {\bf t}, {\bf s}),
\end{gather}
 when expanded in the basis of tensor products of pairs of  power sum symmetric functions $\{p_\mu\}$,
 using the Frobenius character formula (\ref{Frobenius_character}), are generating functions for the weighted  double Hurwitz numbers.
\begin{Theorem}[\cite{GH2}]
\begin{gather}
\tau^{G(z)} ({\bf t}, {\bf s})
= \sum_{d=0}^\infty \sum_{\substack{\mu, \nu, \\ \abs{\mu} = \abs{\nu}}} z^d H^d_G(\mu, \nu) p_\mu({\bf t}) p_\nu({\bf s}),
\\
\tau^{\tilde{G}(z)} ({\bf t}, {\bf s})
= \sum_{d=0}^\infty \sum_{\substack{\mu, \nu, \\ \abs{\mu} = \abs{\nu}}} z^d H^d_{\tilde{G}}(\mu, \nu) p_\mu({\bf t}) p_\nu({\bf s}).
\end{gather}
\end{Theorem}

\section[$2D$ Toda $\tau$-functions as generating functions  for multispecies weighted Hurwitz numbers]{$\boldsymbol{2D}$ Toda $\boldsymbol{\tau}$-functions as generating functions\\  for multispecies weighted Hurwitz numbers}
\label{generating_multispecies_weighted_hurwitz}

 For any choice of  weight generating functions   $G^1(w_1), G^2(w_2), \dots $,  $\tilde{G}^1(z_1), \tilde{G}^2(z_2), \dots $,
 we may form  composites  by  using the product $\prod_{\alpha}G^\alpha(w_\alpha)\prod_{\beta}\tilde{G}^\beta(z_\beta)$
as generating function for multiple weighting types.  The resulting content product coef\/f\/icients
$r_\lambda^{\prod_\alpha G^\alpha(w_\alpha)\prod_{\beta}\tilde{G}^\beta(w_\beta}$ are just the product  $ \prod_{\alpha} r_\lambda^{G^\alpha(w_\alpha)}\prod_{\beta} r_\lambda^{\tilde{G}^\beta(z_\beta)}$
of the individual ones
\begin{gather}
r_\lambda^{G^{\alpha}(w_\alpha)}\deq  \prod_{(i,j)\in \lambda} G^\alpha (w_\alpha(N+ j-i)), \qquad
r_\lambda^{\tilde{G}^{\beta}(z_\beta)}\deq  \prod_{(i,j)\in \lambda} \tilde{G}^\beta(z_\beta(N+ j-i)).
 \end{gather}

We may also include weight factors   in which some or all of the parameters
$(z_1, z_2, \dots)$, $(w_1, w_2, \dots)$ are  repeated in the product. This only  af\/fects the constraints on the sums of the colengths in the  weighted multispecies Hurwitz numbers. (See, e.g., Example~3.3 in~\cite{GH2}, in which the weights
are uniform, but the linear generating function that gives Hurwitz numbers for Belyi
curves and strictly monotonic paths is replaced by a power of the latter, resulting in multiple branch points,
with the total colength f\/ixed, and multimononic paths.)

\subsection[The multiparameter family  of $\tau$-functions $\tau^{G^{(l,m)}({\bf w},{\bf z})}({\bf t},{\bf s})$]{The multiparameter family  of $\boldsymbol{\tau}$-functions $\boldsymbol{\tau^{G^{(l,m)}({\bf w}, {\bf z})}({\bf t}, {\bf s})}$}
\label{multspecies_tau}

The multispecies partitions are divided into two classes: those corresponding to the weight factors of type $G^\alpha(w)$,
 labelled $\big\{  \mu^{(\alpha, u_\alpha)}\big\}_{1\le u_\alpha \le k_\alpha}$
   and those corresponding to  dual type $\tilde{G}^\beta(z)$,
 labelled $ \big\{\nu^{(\beta, v_\beta}\big\}_{1\le v_\beta \le \tilde{k}_\beta} $.
   These are  further subdivided into $l$ ``colours'', or ``species''  for the f\/irst class,
denoted by the label $\alpha =1, \dots , l$ and $m$ in the second, denoted by $\beta=1, \dots , m$.
Any given conf\/iguration
$\big\{ \big\{\mu^{(\alpha, u_\alpha)}\big\}_{1\le u_\alpha \le k_\alpha}, \big\{\nu^{(\beta, v_\beta}\big\}_{1\le v_\beta \le \tilde{k}_\beta}\big\}$
  has $k_\alpha$ elements of colour $\alpha$  in the f\/irst class  and $\tilde{k}_\beta$ elements of colour $\beta$
  in the second class,  for a total of
\begin{gather}
k = \sum_{\alpha=1}^l k_\alpha  + \sum_{\beta=1}^m \tilde{k}_\beta
\end{gather}
partitions.

Denoting the $l+m$  expansion parameters  as
\begin{gather}
{\bf w} = (w_1, \dots, w_l), \qquad {\bf z} = (z_1, \dots, z_m),
\end{gather}
the multispecies weight generating function is formed from the  product
\begin{gather}
G^{(l,m)}({\bf w}, {\bf z}) := \prod_{\alpha=1}^l G^{\alpha}(w_\alpha)  \prod_{\beta=1}^m\tilde{G}^\beta(z_\beta),
\end{gather}
where each  factor has an inf\/inite product representation that is of one of the two types
\begin{gather}
G^{\alpha}(w)  = \prod_{i=1}^\infty \big(1 + c_i^\alpha w\big),  \qquad \alpha =1, \dots , l,\\
\tilde{G}^{\beta}(w)  = \prod_{i=1}^\infty \big(1 - \tilde{c}_i^\beta w\big)^{-1},    \qquad \beta =1, \dots , m,
\end{gather}
for $l+m$ inf\/inite sequences of parameters
\begin{gather}
{\bf c}^\alpha  = \big(c^\alpha_1, c^\alpha_2, \dots\big), \qquad \alpha =1, \dots , l,  \qquad
 \tilde{\bf c}^\beta   = \big(\tilde{c}^\beta_1, \tilde{c}^\beta_2, \dots\big),  \qquad \beta =1, \dots, m.
\end{gather}

 Following the approach developed in \cite{GH2}, we def\/ine an associated element $G_n^{(l,m)}({\bf w}, {\bf z}, \JJ)$ of the center $\Zb(\Cb[S_n])$ of the group algebra $\Cb[S_n]$ by
    \begin{gather}
G_n^{(l,m)}({\bf w}, {\bf z}, \JJ) :=\prod_{a=1}^n\left( \prod_{\alpha=1}^l G^{\alpha}(w_\alpha\JJ_a)\right)
\left(  \prod_{\beta=1}^m\tilde{G}^\beta(z_\beta \JJ_a) \right),
\label{Glm_wzJJ}
\end{gather}
where $\JJ:= (\JJ_1, \dots, \JJ_n)$ are the Jucys--Murphy elements~(\ref{jucys_murphy}) of~$\Cb[S_n]$.
The element \linebreak $G_n^{(l,m)}({\bf w}, {\bf z}, \JJ)$  def\/ines  an endomorphism of~$\Zb(\Cb[S_n])$ under multiplication that is diagonal in the basis  $\{F_\lambda\}$ of $\Zb(\Cb[S_n])$ consisting of the orthogonal idempotents corresponding to irreducible representations, labelled by partitions~$\lambda$ of~$n$
\begin{gather}
G_n^{(l,m)}({\bf w}, {\bf z}, \JJ) F_\lambda =r_\lambda^{G^{(l,m)}({\bf w}, {\bf z}) }F_\lambda.
\label{central_G_nlmwz_generator_eigenvalue}
\end{gather}
The eigenvalues are
\begin{gather}
r_\lambda^{G^{(l,m)}({\bf w}, {\bf z}) }= \prod_{\alpha=1}^lr_\lambda^{G^\alpha(w_\alpha)}\prod_{\beta=1}^m
r_\lambda^{\tilde{G}^\beta(z_\beta)},
\end{gather}
where, as before,
\begin{gather}
   r_\lambda^{G^\alpha(w_\alpha) } := \prod_{(ij)\in \lambda} G^\alpha(w_\alpha(j-i)),  \qquad  r_\lambda^{\tilde{G}(z_\beta)}
    := \prod_{(ij)\in \lambda}  \tilde{G}^\beta(z_\beta(j-i)).
\end{gather}

The associated hypergeometric $\tau$-function is
\begin{gather}
\tau^{G^{(l,m)}({\bf w}, {\bf z})}({\bf t}, {\bf s})    = \sum_{\lambda}
r_\lambda^{G^{(l,m)}({\bf w}, {\bf z}) }    s_\lambda({\bf t}) s_\lambda({\bf s}) .
\label{multi_colour_hybrid_weighted_tau}
\end{gather}

\subsection{Multispecies geometric weighted Hurwitz numbers}
\label{multispecies_geometric_hurwitz}

 The weight assigned to a multispecies $n$-sheeted branched covering
 of the Riemann sphere with a pair of branch points at, say, $(0, \infty)$ having
 ramif\/ication prof\/iles $(\mu, \nu)$ and a further
\begin{gather}
 \sum_{\alpha=1}^l k_\alpha + \sum_{\beta=1}^l \tilde{k}_\beta
\end{gather}
 ``coloured''  branch  points of two classes:
 $\big\{\mu^{(\alpha, u_\alpha)}\big\}_{\substack{\alpha =1, \dots, l \\ u_\alpha = 1, \dots, k_\alpha}}$ (``class~I'' ) and
 $\big\{\nu^{(\beta, v_\beta)}\big\}_{\substack{\beta =1, \dots, l \\ v_\beta = 1, \dots, \tilde{k}_\beta}}$ (``class~II'')
 and $l+m$ colours $\{\alpha = 1, \dots, l\}$, $\{\beta = 1, \dots, m\}$,
with $\{k_\alpha\}_{\alpha=1, \dots, l}$, $\{\tilde{k}_\beta\}_{\beta=1, \dots, m}$
points of the various colours (or ``species''),  is def\/ined to be the product of those for single species
 \begin{gather}
W_{G^{(l,m)}} \big(\big\{\mu^{(\alpha, u_\alpha)}\big\}, \big\{\nu^{(\beta, v_\beta}\big\}\big)  =
\prod_{\alpha=1}^l  m_{\lambda^{(\alpha)}} \big({\bf c}^{(\alpha)}\big)
\prod_{\beta=1}^m m_{\tilde{\lambda}^{(\beta)}} \big(\tilde{\bf c}^\beta\big).
 \end{gather}
 Here  the partitions $\big\{\lambda^{(\alpha)}\big\}_{\alpha=1, \dots, l}$,
 and $\big\{\tilde{\lambda}^{(\beta)}_{\beta=1, \dots, m}\big\}$  have parts  equal to the
 colengths $\ell^*\big(\mu^{(\alpha, u_\alpha)}\big)$ and $\ell^*\big(\mu^{(\beta, v_\beta)}\big)$, for $\lambda^{(\alpha)}$ and $\tilde{\lambda}^{(\beta)}$  respectively.
Their lengths are
\begin{gather}
 \ell\big(\lambda^{(\alpha)}\big) = k_\alpha, \qquad   \ell\big(\tilde{\lambda}^{(\beta)}\big) = \tilde{k}_\beta
\end{gather}
 and their weights
\begin{gather}
 \big|\lambda^{(\alpha)})\big| = d_\alpha, \qquad   \big|\tilde{\lambda}^{(\beta)}\big| = \tilde{d}_\beta
\end{gather}
are equal to the specif\/ied total colengths
 \begin{gather}
 {\bf d}  =(d_1, \dots , d_l),  \qquad \tilde{\bf d} =(d_1, \dots , \tilde{d}_m),  \\
 d_\alpha  = \sum_{u_\alpha=1}^{k_\alpha} \ell^*\big(\mu^{(\alpha, u_\alpha)}\big), \qquad  \tilde{d}_\beta = \sum_{v_\beta=1}^{\tilde{k}_\beta} \ell^*\big(\nu^{(\beta v_\beta)}\big).
\end{gather}

The geometrically def\/ined  multispecies weighted Hurwitz numbers are
 \begin{gather}
H_{G^{(l,m)}}^{({\bf d}, \tilde{\bf d})}(\mu, \nu)
 := \sum_{k_1, \dots, k_l}  \sum_{\tilde{k}_1, \dots , \tilde{k}_m}   \sideset{}{'}\sum_{\substack{\{\mu^{(\alpha, u_\alpha)}\} \\ |\mu^{(\alpha, u_\alpha)}|=n \\
 \sum\limits_{u_\alpha=1}^{k_\alpha} \ell^*(\mu^{(\alpha, u_\alpha)}) = d_\alpha}}
 \sideset{}{'}\sum_{\substack{\{\nu^{(\beta, v_\beta)}\} \\ |\nu^{(\beta, u_\beta)}|=n \\
 \sum\limits_{v_\beta=1}^{\tilde{k}_\beta} \ell^*(\nu^{(\beta, v_\beta)}) = \tilde{d}_\beta}}\nonumber\\
 \hphantom{H_{G^{(l,m)}}^{({\bf d}, \tilde{\bf d})}(\mu, \nu) :=}{}
  \times W_{G^{(l,m)}} \big(\big\{\mu^{(\alpha, u_\alpha)}\big\}, \big\{\nu^{(\beta, v_\beta)}\big\}\big)
H\big(\big\{\mu^{(\alpha, u_\alpha)}\big\}, \big\{\nu^{(\beta, v_\beta)}\big\}, \mu, \nu\big),
\end{gather}
which gives the weighted sum of the Hurwitz numbers  of $n$-sheeted branched coverings with  $l+m$ branch points of type
$\big\{ \big\{\mu^{(\alpha, u_\alpha)}\big\}_{1\le u_\alpha \le k_\alpha}, \big\{\nu^{(\beta, v_\beta}\big\}_{1\le v_\beta \le \tilde{k}_\beta}\big\}$
and $(\mu, \nu)$ at $(0, \infty)$,

Substituting the Frobenius--Schur formula~(\ref{Frobenius_Schur_Hurwitz}) and the
Frobenius character formula~(\ref{Frobenius_character})
into~(\ref{multi_colour_hybrid_weighted_tau}), it follows that  $\tau^{G^{(l,m)}({\bf w}, {\bf z})}({\bf t}, {\bf s})   $
is the generating function for  $H_{G^{(l,m)}}^{({\bf d}, \tilde{\bf d})}(\mu, \nu)  $.
Using multi-index notion to denote
\begin{gather}
\prod_{\alpha=1}^l w_\alpha^{d_\alpha} \prod_{\beta=1}^m z_\beta^{\tilde{d}_\beta} =:  {\bf w}^{\bf d} {\bf z}^{\tilde{\bf d}},
\end{gather}
  we then have
\begin{Theorem}
\begin{gather}
\tau^{G^{(l,m)}({\bf w}, {\bf z})}({\bf t}, {\bf s})   =\sum_{{\bf d} \in \Nb}  {\bf w}^{\bf d} \sum_{\tilde{\bf d} \in \Nb }
{\bf z}^{\tilde{\bf d}}\sum_{\mu, \nu}
H_{G^{(l,m)}}^{({\bf d}, \tilde{\bf d})}(\mu, \nu)  p_\mu({\bf t}) p_\nu({\bf s}).
 \label{multispecies_geometric_tau}
 \end{gather}
\end{Theorem}

 \begin{proof}
 This follows, as in the single species case, by combining the eigenvalue formula~(\ref{central_G_nlmwz_generator_eigenvalue})
 with the Frobenius character formula.
 \end{proof}

\subsection{Multispecies combinatorial weighted Hurwitz numbers}
\label{multispecies_combinatorial_hurwitz}

The combinatorial multispecies weighted Hurwitz number $F_{G^{(l,m)}}^{({\bf d}, \tilde{\bf d})}(\mu, \nu) $ is
 def\/ined as follows.
 Let~$D_n$ be the number of partitions of  $n$ and let ${\bf F}_{G^\alpha}^{d_\alpha}$ and
${\bf F}_{\tilde{G}^\beta}^{\tilde{d}_\beta}$ denote  the $D_n \times D_n$ matrices whose elements
 are $F^{d_\alpha}_{G^\alpha}(\mu, \nu)$ and $F^{\tilde{d}_\beta}_{\tilde{G}^\beta}(\mu, \nu)$,
  respectively, as def\/ined in~(\ref{Fd_G_mu_nu}),~(\ref{Fd_tilde_G_mu_nu}),
  \begin{gather}
F^d_{G^\alpha}(\mu, \nu)  := {1\over n!} \sum_\lambda e_\lambda\big({\bf c}^\alpha\big) m^\lambda_{\mu \nu},
 \qquad
F^d_{\tilde{G}^\beta}(\mu, \nu)   := {1\over n!} \sum_\lambda h_\lambda\big(\tilde{\bf c}^\beta\big) m^\lambda_{\mu \nu},
\end{gather}
 for each generating function $G^\alpha(w_\alpha)$ or $\tilde{G}^\beta(z_\beta)$.
Since the central elements
 $\big\{G_n^\alpha(w_\alpha, \JJ)$,\linebreak $\tilde{G}_n^\beta (z_\beta, \JJ)\big\}$
all commute, it follows that so do the matrices $\big\{{\bf F}^{d_\alpha}_{G^\alpha},
{\bf F}^{\tilde{d}_\beta}_{\tilde{G}^\beta}\big\}$.  Denoting their product, in any order,
\begin{gather}
{\bf F}^{({\bf d}, \tilde{\bf d})}_{G^{(l,m)}}
:= \prod_{\alpha=1}^l {\bf F}_{G^\alpha}^{d_\alpha} \prod_{\beta=1}^m {\bf F}_{\tilde{G}^\beta}^{\tilde{d}_\beta},
\end{gather}
the $(\mu, \nu)$  matrix element $F^{({\bf d}, \tilde{\bf d})}_{G^{(l, m)}}(\mu, \nu)$ is the combinatorial multispecies
weighted Hurwitz number.

The combinatorial meaning of $F^{({\bf d}, \tilde{\bf d})}_{G^{(l, m)}}(\mu, \nu)$ is as follows.
Let
\begin{gather}
d := \sum_{\alpha=1}^l d_\alpha + \sum_{\beta=1}^m \tilde{d}_\beta.
\end{gather}
Then  $F^{({\bf d}, \tilde{\bf d})}_{G^{(l,m)}}  (\mu, \nu)$, may be interpreted as the weighted sum
over all sequences  of $d$ step paths in the Cayley graph from an element $h \in \cyc(\mu)$ in
the conjugacy class of cycle  type $\mu$ to one $(a_db_d) \cdots (a_1b_1)h \in \cyc(\nu)$ in the
class  $\cyc(\nu)$, in which the transpositions appearing are subdivided into subsets consisting of $(d_1, \dots, d_l, \tilde{d}_1, \dots , \tilde{d}_m)$  transpositions  in all $ {d! \over  \prod\limits_{\alpha=1}^l d_\alpha !    \prod\limits_{\beta=1}^l \tilde{d}_\beta !}$ possible ways. All paths are divided into equivalence classes, according to their {\em multisignatures}  $\big\{\lambda^{(\alpha)}, \tilde{\lambda}^{(\beta)}\big\}_{\substack{\alpha =1, \dots, l \\ \beta =1, \dots , m}}$.
These consist  of a partition of the $d$  steps into $l +m$ parts, each of which is a subsequence assigned a~``colour'' and a~``class'' with~$l$ of them of the f\/irst class and~$m$ of  the second.
The number of partitions of f\/irst class with colour~$\alpha$ is~$d_\alpha$  while the number of
second class with colour~$\beta$ is $\tilde{d}_\beta$. The partitions  $\lambda^{(\alpha)}$ of  weights $d_\alpha$
are def\/ined to have parts  $\big\{ \lambda^{(\alpha)}_{u_\alpha}\big\}_{u_\alpha=1, \dots, k_\alpha}$ equal to the number of transpositions appearing within that subsequence having the same second element, and similarly for
$\big\{\tilde{\lambda}^{(\beta)}_{v_\beta}\big\}_{v_\beta =1, \dots , \tilde{k}_\beta}$
with
\begin{gather}
k_\alpha= \ell\big( \lambda^{(\alpha)}\big), \qquad\tilde{k}_\beta= \ell\big( \tilde{\lambda}^{(\beta)}\big)
\end{gather}
the number of such parts.

The weight given to any such multisignatured path is  the product
\begin{gather}
\prod_{\alpha=1}^l e_{\lambda^{(\alpha)}} \big({\bf c}^\alpha\big)
\prod_{\beta =1}^m h_{\tilde{\lambda}^{(\beta)}} \big(\tilde{{\bf c}}^\beta\big)
\end{gather}
of the weights along each segment, and
 $F^{({\bf d}, \tilde{\bf d})}_{G^{(l,m)}}  (\mu, \nu)$  is the sum of these, each  multiplied by the number
of elements of the equivalence class of paths with the given multisignature.

The multispecies generalization of (\ref{Fd_G_mu_nu_Hd}) is  equality of the geometric and combinatorial Hurwitz numbers:
\begin{Theorem}
\begin{gather}
F^{({\bf d}, \tilde{\bf d})}_{G^{(l, m)}}  (\mu, \nu)   = H^{({\bf d}, \tilde{\bf d})}_{G(l, m)}  (\mu, \nu),
 \qquad
F^{({\bf d}, \tilde{\bf d})}_{\tilde{G}(l, m)}  (\mu, \nu)  = H^{({\bf d}, \tilde{\bf d})}_{\tilde{G}(l, m)}  (\mu, \nu) .
\end{gather}
\end{Theorem}

\begin{proof}
Applying the central element $G_n^{(l,m)}({\bf w}, {\bf z}, \JJ)$ def\/ined in (\ref{Glm_wzJJ})
to the cycle sum $C_\mu$  and applying~(\ref{Gn_z_JJ}) for each factor in the product gives
\begin{gather}
G_n^{(l,m)}({\bf w}, {\bf z}, \JJ)  C_\mu
 = \sum_{\nu,\,  |\nu|=|\mu|} F^{({\bf d}, \tilde{\bf d})}_{G^{(l,m)}}  (\mu, \nu)   C_\nu
 =\sum_{\nu, \, |\nu|=|\mu|} H^{({\bf d}, \tilde{\bf d})}_{G^{(l,m)}}  (\mu, \nu)   C_\nu,
\end{gather}
 and the similar formula for $\tilde{G}$.
Equation (\ref{multispecies_geometric_tau}), together with the Frobenius character formula,  shows that the
generating $\tau$-function can be expressed as
\begin{gather}
\tau^{G^{(l,m)}({\bf w}, {\bf z})}({\bf t}, {\bf s})   =\sum_{{\bf d} \in \Nb}  {\bf w}^{\bf d} \sum_{\tilde{\bf d} \in \Nb }
{\bf z}^{\tilde{\bf d}}\sum_{\mu, \nu}
 F^{({\bf d}, \tilde{\bf d})}_{G^{(l,m)}}  (\mu, \nu) p_\mu({\bf t}) p_\nu({\bf s}).  
 \end{gather}
Comparing this with equation~(\ref{multispecies_geometric_tau}) proves the result.
\end{proof}

\section{Multispecies quantum Hurwitz numbers}
\label{multispecies_quantum_hurwitz}

\subsection{Quantum Hurwitz numbers}
\label{quantum_hurwitz}

 Amongst the examples  of weighted Hurwitz numbers studied in~\cite{GH2},  four special classes were introduced
 in which the generating functions $G(z)$, $\tilde{G}(z)$ were chosen as a variant of the quantum dilogarithm~\cite{FK}. This meant that the parameters~$\{c_i\}$ were chosen as powers of a quantum deformation parameter~$q$.  As shown in~\cite{GH2} , a suitable interpretation of the parameter~$q$ in terms of Planck's constant $\hbar$ and Boltzmann factors
for a Bosonic gas with linear  energy spectrum leads to a relation between  the  resulting weighted counting
of branched cover and the energy distribution for a Bosonic gas, which further justif\/ies terming these ``quantum'' Hurwitz numbers.

In the f\/irst case, the weight generating function  is
\begin{gather}
E(q,z)  := \prod_{i=0}^\infty \big(1+q^i z\big) = 1 +\sum_{i=1}^\infty E_i(q)z^i , \qquad
E_i(q) := \left(\prod_{j=1}^i {q^{j-1} \over 1 - q^j}\right),
\end{gather}
and hence the parameters $c_i$ are identif\/ied  as $\{c_i := q^{i-1}\}_{i\in \Nb^+}$.
The second is a slight modi\-f\/i\-ca\-tion of this, with weight generating function
\begin{gather}
E'(q,z)  := \prod_{i=1}^\infty \big(1+q^i z\big) = 1 +\sum_{i=1}^\infty E'_i(q)z^i , \qquad
E'_i(q) := \left(\prod_{j=1}^i {q^i \over 1 - q^j}\right),
\end{gather}
i.e.,  the zero power $q^0$ is omitted, and $\{c_i := q^{i}\}_{i\in \Nb^+}$.

The third case is based on the  weight generating function
\begin{gather}
H(q,z)  := \prod_{i=0}^\infty\big(1-q^i z\big)^{-1} =1 +\sum_{i=1}^\infty H_i(q)z^i , \qquad
H_i(q) := \left(\prod_{j=1}^i {1 \over 1 - q^j}\right),
\end{gather}
and hence is the dual of the f\/irst case,  with $\{ \tilde{c}_i := q^{i-1}\}_{i\in \Nb^+}$.
The f\/inal case is a hybrid, formed from the product  of the f\/irst and third
for two distinct quantum deformation parameters~$q$ and~$p$, with weight generating function
\begin{gather}
Q(q, p, z)   := \prod_{k=0}^\infty \big(1+ q^k z\big) \big(1- p^k z\big)^{-1} = \sum_{i=0}^\infty Q_i(q,p)z^i, \\
Q_i(q,p)   := \sum_{m=0}^i q^{\frac{1}{2}m(m-1)} \left(\prod_{j=1}^m \big(1-q^j\big) \prod_{j=1}^{i -m}(1-p^j)\right)^{-1}\!,
\qquad Q_\lambda(q,p) =\prod_{i=1}^{\ell(\lambda)} Q_{\lambda_i}(q,p).\!\!\!
\end{gather}
These are all expressible as exponentials  of the quantum dilogarithm function
\begin{gather}
\Li_2(q, z)  \deq \sum_{k=1}^\infty \frac{z^k}{k (1- q^k)},\\
 E(q, z)  = e^{-\Li_2(q, -z)}, \qquad   E'(q, z) = (1+z)^{-1}e^{-\Li_2(q, -z)}, \\
 H(q, z)    = e^{\Li_2(q, z)}, \qquad   Q(q, p, z)  = e^{\Li_2(p, z)-\Li_2(q, -z)}.
\end{gather}

   The  content product formulae for the f\/irst and third of these are
\begin{gather}
r_\lambda^{E(q,z)}(N)    := \prod_{(ij) \in \lambda} E(q, (N+j-i)z), \qquad
r_\lambda^{H(q,z)}(N)    := \prod_{(ij) \in \lambda} H(q, (N+ j-i)z).
\end{gather}
The associated hypergeometric $2D$ Toda $\tau$-functions have
diagonal double Schur function expansions with these as coef\/f\/icients
\begin{gather}
\tau^{E(q,z)}(N, {\bf t}, {\bf s})  =  \sum_{\lambda} r_\lambda^{E(q,z)} (N) S_\lambda({\bf t}) S_\lambda({\bf s}), \\
\tau^{H(q,z)}(N, {\bf t}, {\bf s})  =  \sum_{\lambda} r_\lambda^{H(q,z)} (N) S_\lambda({\bf t}) S_\lambda({\bf s)}).
\end{gather}

Using the Frobenius character formula~(\ref{Frobenius_character}),
and setting $N=0$, they may be rewritten as double expansions
 in the power sum symmetric functions~\cite{GH2}
 \begin{gather}
\tau^{E(q,z)}({\bf t}, {\bf s}) := \tau^{E(q,z)}(0, {\bf t}, {\bf s})=\sum_{d=0}^\infty z^d\sum_{\mu, \nu, \, |\mu|=|\nu|}  H^d_{E(q)}(\mu, \nu)   p_\mu({\bf t}) p_\nu({\bf s}), \\ 
\tau^{H(q,z)}({\bf t}, {\bf s})   :=\tau^{H(q,z)}(0, {\bf t}, {\bf s}) = \sum_{d=0}^\infty z^d \sum_{\mu, \nu, \, |\mu|=|\nu|}  H^d_{H(q)}(\mu, \nu)  p_\mu({\bf t}) p_\nu({\bf s}).
\label{quantum_hurwitz_expansion}
\end{gather}
The coef\/f\/icients are the corresponding quantum Hurwitz numbers $H^d_{E(q)}(\mu, \nu)$, $H^d_{H(q)}(\mu, \nu)$,
 which count weighted $n$-sheeted branched coverings of the Riemann sphere,   def\/ined by
\begin{gather}
H^d_{E(q)}(\mu, \nu) := \sum_{k=0}^\infty  \sideset{}{'}\sum_{\mu^{(1)}, \dots, \mu^{(k)} \atop \sum\limits_{i=1}^k \ell^*(\mu^{(i)} )= d}
 W_{E(q)}\big(\mu^{(1)}, \dots , \mu^{(k)}\big)  H\big(\mu^{(1)}, \dots , \mu^{(k)}, \mu, \nu\big) ,\\
H^d_{H(q)}(\mu, \nu) := \sum_{k=0}^\infty   (-1)^{k+d} \sideset{}{'}\sum_{\mu^{(1)}, \dots, \mu^{(k)} \atop \sum\limits_{i=1}^k \ell^*(\mu^{(i)} )= d}
 W_{H(q)}\big(\mu^{(1)}, \dots , \mu^{(k)}\big)  H\big(\mu^{(1)}, \dots , \mu^{(k)}, \mu, \nu\big),
\end{gather}
where the weightings for such covers with $k$ additional branch points are~\cite{GH2}
\begin{gather}
W_{E(q)} \big(\mu^{(1)}, \dots, \mu^{(k)}\big)
 \nonumber\\
 \qquad {} ={1\over\abs{\aut(\lambda)}} \sum_{\sigma\in S_k}   {q^{(k-1) \ell^*(\mu^{(1)})} \cdots  q^{ \ell^*(\mu^{(k-1)})} \over
 \big(1- q^{\ell^*(\mu^{(\sigma(1))})} \big)  \cdots \big(1- q^{\ell^*(\mu^{(\sigma(1))})}\big) \cdots q^{\ell^*(\mu^{(\sigma(k))})}},\\
 W_{H(q)} \big(\mu^{(1)}, \dots, \mu^{(k)}\big)\nonumber\\
 \qquad{} = {(-1)^{\ell^*(\lambda)}\over\abs{\aut(\lambda)}}\sum_{\sigma\in S_k}   {1 \over
 \big(1- q^{\ell^*(\mu^{(\sigma(1))})} \big)  \cdots \big(1- q^{\ell^*(\mu^{(\sigma(1))})}\big) \cdots q^{\ell^*(\mu^{(\sigma(k))})}}.
\end{gather}
   Here, as in (\ref{WG_weight})  and (\ref{dual_WG_weight}),  $\lambda$ is the partition of  length  $\ell(\lambda)=k$ whose parts are $\big\{\ell^*\big(\mu^{(i)}\big)\big\}_{i=1, \dots, k}$,
and $\abs{\aut(\lambda)} $ is the order of the automorphism group of $\lambda$ .
 The sum $\sum'_{\mu^{(1)}, \dots, \mu^{(k)} \atop \sum\limits_{i=1}^k \ell^*(\mu^{(i)} )= d} $ is over all $k$-tuples of partitions  having nontrivial ramif\/ication prof\/iles that satisfy  the constraint $\sum\limits_{i=1}^k \ell^*(\mu^{(i)} )= d$,  and $H\big(\mu^{(1)}, \dots , \mu^{(k)}, \mu, \nu\big)$ is the number of branched $n$-sheeted coverings, up to isomorphism,  having $k+2$ branch points with ramif\/ication prof\/iles $\big(\mu^{(1)}, \dots , \mu^{(k)}, \mu, \nu\big)$.

These thus count weighted covers with a pair of branch points, say  $(0, \infty)$, having
ramif\/ication prof\/iles of type $(\mu, \nu)$  and an  arbitrary number
of further  branch points, whose prof\/iles $(\mu^{(1)}$, $\dots , \mu^{(k)})$ are constrained  only  by the requirement that the sum of the colengths,
which is related to the genus by the Riemann--Hurwitz formula
\begin{gather}
\sum_{i=1}^k \ell^*(\mu^{(i)} )= 2g -2  +\ell(\mu) +\ell(\nu) =d,
\end{gather}
be f\/ixed to equal $d$.

The combinatorial interpretation of the quantum Hurwitz numbers
 $ F^d_{E(q)}(\mu, \nu)$ and $ F^d_{H(q)}(\mu, \nu)$  appearing in~(\ref{quantum_hurwitz_expansion}) is as
 follows. Let $(a_1b_1) \cdots (a_d b_d)$ be a product of $d$ transpositions $(a_i b_i) \in S_n$ in the symmetric
group  $S_n$ with $a_i < b_i$, $i=1, \dots, d$.
If $h \in  S_n$ is in the conjugacy class $ \cyc(\mu)\ss S_n$, we may view the successive steps in the product
\begin{gather}
(a_1b_1) \cdots (a_d b_d) h
\end{gather}
as a path in the Cayley graph generated by all transpositions, whose  {\it signature}
is the parti\-tion~$\lambda$ of~$d$, $|\lambda|=d$, whose parts~$\lambda_i$ consist
of the number of transpositions~$(a_i b_i)$ sharing the same second element. If we further
require that the ones with equal second elements be grouped together into consecutive subsequences in which these second elements
are constant,  with the consecutive subsequences  strictly increasing in their second elements, then the number~$\tilde{N}_\lambda$
of elements with signature~$\lambda$ is related to the number~$N_\lambda$ of such ordered sequences by
\begin{gather}
\tilde{N}_\lambda = {|\lambda|! \over \prod\limits_{i=1}^{\ell(\lambda)} \lambda_i !} N_\lambda.
\end{gather}
Denote the number of such paths from the conjugacy class of cycle type $\cyc(\mu)$ to the one of type $\cyc(\nu)$
having signature $\lambda$ as $\tilde{m}^\lambda_{\mu \nu}$, and the number of ordered sequences
of type $\lambda$ as  $m^\lambda_{\mu \nu}$. Thus
\begin{gather}
\tilde{m}^\lambda_{\mu \nu} =  {|\lambda|! \over \prod\limits_{i=1}^{\ell(\lambda)} \lambda_i !}  m^\lambda_{\mu \nu}.
\end{gather}

For all paths of signature $\lambda$ we  assign the weights
\begin{gather}
E_\lambda(q)  :=  \prod_{i=1}^{\ell(\lambda)}E_{\lambda_i}(q)
 =\prod_{i=1}^{\ell(\lambda)}{ q^{{1\over 2}\lambda_i(\lambda_i -1)} \over \prod\limits_{j=1}^{\lambda_i} (1-q^j)} , \qquad
H_\lambda(q)  :=  \prod_{i=1}^{\ell(\lambda)} H_{\lambda_i}(q)
 =\prod_{i=1}^{\ell(\lambda)}{1\over \prod\limits_{j=1}^{\lambda_i} (1-q^j)}
\end{gather}
to paths in the Cayley graph, and obtain the pair of corresponding combinatorial weighted Hurwitz numbers
\begin{gather}
F^d_{E(q)} (\mu, \nu)  := {1\over n!} \sum_{\lambda, \, |\lambda|=d} E_\lambda(q) m^\lambda_{\mu \nu},
\label{Fd_Eq}
\\
F^d_{H(p)} (\mu, \nu)  := {1\over n!} \sum_{\lambda, \, |\lambda|=d} H_\lambda(q) m^\lambda_{\mu \nu},
\label{Fd_Hq}
\end{gather}
that give the weighted enumeration of paths, using the weighting factors $E_\lambda(q)$ and $H_\lambda(q) $
respectively for all paths of signature~$\lambda$.

As shown in general in~\cite{GH2}, the enumerative geometrical and combinatorial def\/initions
of these quantum weighted Hurwitz numbers coincide:
\begin{gather}
H^d_{E(q)}(\mu, \nu)  = F^d_{E(q)}(\mu, \nu),  \qquad H^d_{H(q)}(\mu, \nu) = F^d_{H(q)}(\mu, \nu).
\label{Hd_equals_Fd}
\end{gather}
A similar result holds for weights generated by the function $E'(q,z)$, with corresponding quantum Hurwitz numbers
def\/ined by
\begin{gather}
H^d_{E'(q)}(\mu, \nu) := \sum_{k=0}^\infty  \sideset{}{'}\sum_{\mu^{(1)}, \dots, \mu^{(k)} \atop \sum\limits_{i=1}^k \ell^*(\mu^{(i)} )= d}
 W_{E'(q)}\big(\mu^{(1)}, \dots , \mu^{(k)}\big)  H\big(\mu^{(1)}, \dots , \mu^{(k)}, \mu, \nu\big) ,
\end{gather}
where the weights $W_{E'(q)} (\mu^{(1)}, \dots, \mu^{(k)})$ are given by
\begin{gather}
W_{E'(q)} \big(\mu^{(1)}, \dots, \mu^{(k)}\big)
:={1\over k!} \sum_{\sigma\in S_k}   {q^{(k) \ell^*(\mu^{(1)})} \cdots  q^{ \ell^*(\mu^{(k)})} \over
 \big(1- q^{\ell^*(\mu^{(\sigma(1))})} \big)  \cdots \big(1- q^{\ell^*(\mu^{(\sigma(1))})}\big) \cdots q^{\ell^*(\mu^{(\sigma(k))})}} \\
\hphantom{W_{E'(q)} \big(\mu^{(1)}, \dots, \mu^{(k)}\big) }{}
:={1\over k!} \sum_{\sigma\in S_k}   {1 \over
 \big(q^{-\ell^*(\mu^{(\sigma(1))})}-1  \big)  \cdots \big(q^{-\ell^*(\mu^{(\sigma(1))})} \cdots q^{-\ell^*(\mu^{(\sigma(k))})}-1 \big)}.
 \nonumber
 \end{gather}

Choosing $q$ as a positive real number, parametrizing it as
\begin{gather}
q= e^{- \beta \hbar \omega_0},  \qquad \beta={1\over kT}
\end{gather}
and identifying the energy levels $\epsilon_k$   as those for a Bose gas with  linear spectrum in the
integers, as for a $1-D$ harmonic oscillator
\begin{gather}
\epsilon_k := k \hbar \omega_0, \qquad k\in \Nb,
\end{gather}
it follows that if we assign the energy
\begin{gather}
\epsilon(\mu) := \epsilon_{\ell^*(\mu)} = \ell^*(\mu)\hbar \omega
\end{gather}
to each branch point with ramif\/ication prof\/ile of type~$\mu$, it contributes a factor
\begin{gather}
n(\mu) ={1 \over e^{\beta\epsilon(\mu)} -1}
\end{gather}
to the weighting distributions, the same as that for a bosonic gas.

\subsection[The multiparameter family  of $\tau$-functions $\tau^{Q({\bf q}; {\bf w};  {\bf p}, {\bf z})} (N,{\bf t}, {\bf s})$]{The multiparameter family  of $\boldsymbol{\tau}$-functions $\boldsymbol{\tau^{Q({\bf q}; {\bf w};  {\bf p}, {\bf z})} (N,{\bf t}, {\bf s})}$}
\label{multispecies_quantum_tau}

We now consider the multiparameter family of weight generating functions $Q({\bf q}, {\bf w}; {\bf p}, {\bf z}) $ obtained by taking the product of any number of the generating functions $E(q_i, z_i)$ and $H(p_j, w_j)$
for distinct sets  of generating function parameters ${\bf w} = (w_1, \dots , w_l)$, $ {\bf z} = (z_1, \dots , z_m)$,
and   quantum parameters $ {\bf q} = (q_1, \dots , q_l)$, ${\bf p} = (p_1, \dots , p_m)$
\begin{gather}
Q({\bf q}, {\bf w}; {\bf p}, {\bf z}) := \prod_{\alpha=1}^l E(q_\alpha, w_\alpha)  \prod_{\beta=1}^m H(p_\beta, z_\beta).
\end{gather}

   Following the approach developed in \cite{GH2}, we def\/ine an associated element of the
    center $\Zb(\Cb[S_n])$ of the group algebra $\Cb[S_n]$ by
   \begin{gather}
Q_n({\bf q}, {\bf w}; {\bf p}, {\bf z}, \JJ) := \prod_{a=1}^n Q({\bf q},  \JJ_a {\bf w}; {\bf p},  \JJ_a {\bf z}) ,
   \end{gather}
where $\JJ:= (\JJ_1, \dots, \JJ_n)$ are again the Jucys--Murphy elements~(\ref{jucys_murphy}) of~$\Cb[S_n]$.
 The element $Q_n({\bf q}, {\bf w}; {\bf p}, {\bf z}, \JJ)$  def\/ines
an endomorphism of $\Zb(\Cb[S_n])$ under multiplication, which is diagonal in the basis~$\{F_\lambda\}$
of $\Zb(\Cb[S_n])$ consisting of the orthogonal idempotents, corresponding to irreducible representations, labelled by partitions~$\lambda$ of~$n$
\begin{gather}
Q_n({\bf q}, {\bf w}; {\bf p}, {\bf z}, \JJ) F_\lambda = r_\lambda^{Q({\bf q},  {\bf w}; {\bf p},  {\bf z})} F_\lambda,
\label{central_qpzw_generator}
 \end{gather}
 where
 \begin{gather}
r_\lambda^{Q({\bf q},  {\bf w}; {\bf p},  {\bf z})} =  \prod _{\alpha=1}^l r_\lambda^{E(q_\alpha)} (w_\alpha)    \prod_{\beta=1}^m r_\lambda^{H(p_\beta)} (z_\beta).
\end{gather}

More generally, def\/ining
 \begin{gather}
r_\lambda^{Q({\bf q},  {\bf w}; {\bf p},  {\bf z})} (N) =  \prod _{\alpha=1}^l r_\lambda^{E(q_\alpha, w_\alpha)} (N)    \prod_{\beta=1}^m
 r_\lambda^{H(p_\beta, , z_\beta)} (N),
\end{gather}
where
\begin{gather}
r_\lambda^{E(q_\alpha, w_\alpha)} (N)   := \prod_{(ij) \in \lambda} E(q_\alpha, (N+ j-i)w_\alpha), \\
r_\lambda^{H(p_\beta, z_\beta)}(N)    := \prod_{(ij) \in \lambda} H(p_\beta, (N+ j-i)z_\beta),
\end{gather}
we have
\begin{gather}
r_\lambda^{Q({\bf q},  {\bf z}; {\bf p},  {\bf w})}  =  r_\lambda^{Q({\bf q},  {\bf z}; {\bf p},  {\bf w})} (0) .
\end{gather}
The double Schur function series
\begin{gather}
\tau^{Q({\bf q}, {\bf z}; {\bf p},  {\bf w})}(N, {\bf t}, {\bf s}) := \sum_\lambda  r_\lambda^{Q({\bf q},  {\bf z}; {\bf p},  {\bf w})} (N)
  S_\lambda({\bf t}) S_\lambda({\bf s})
\label{tau_qz_pw}
\end{gather}
then def\/ines a family of $2D$ Toda $\tau$-functions of hypergeometric type.

\subsection{Multispecies geometric quantum Hurwitz numbers}
\label{multispecies_quantum_geometric_hurwitz}

We now consider coverings in which the branch points are divided, as above,  into two dif\/ferent classes,
$\big\{\mu^{(\alpha, u_\alpha)}\big\}_{\alpha=1, \dots, l; \,  u_\alpha = 1 ,\dots , k_\alpha}$ and $\big\{\nu^{(\beta, v_\beta)}\big\}_{\beta=1, \dots, m; \, v_\beta= 1, \dots , \tilde{k}_\beta}$,
corresponding to weight ge\-ne\-rating functions of type $E(q_\alpha)$ and $H(p_\beta)$ respectively, each of which is
further divided into~$l$ and~$m$ distinct species (or ``colours''), of which there are~$\big\{k^\alpha\big\}$ and
$\big\{\tilde{k}^\beta\big\}$ branch points of types~$E$ and~$H$ and colours~$\alpha$ and~$\beta$ respectively. The weighted number of such coverings
 with specif\/ied total colengths    ${\bf d} =(d_1, \dots , d_l)$,  $\tilde{\bf d} =(d_1, \dots , \tilde{d}_m)$,  $d_\alpha, \tilde{d}_\beta \in \Nb$
 \begin{gather}
 d_\alpha = \sum_{u_\alpha=1}^{k_\alpha} \ell^*\big(\mu^{(\alpha, u_\alpha)}\big), \qquad  \tilde{d}_\beta = \sum_{v_\beta=1}^{\tilde{k}_\beta} \ell^*\big(\nu^{(\beta v_\beta)}\big)
 \end{gather}
  for each  class and colour is  the multispecies quantum Hurwitz number
\begin{gather}
 H^{({\bf d}, \tilde{\bf d})}_{({\bf q}, {\bf p})}(\mu, \nu)  :=
  \sum_{\{k_\alpha=1; \, \tilde{k}_\beta=1\} \atop \alpha=1, \dots, l; \, \beta=1, \dots, m}^\infty
 \sum_{\big\{\mu^{(\alpha, u_\alpha)};\,  \nu^{(\beta, v_\beta)}\big\}\atop \sum\limits_{u_\alpha=1}^{k_\alpha}\ell^*\big(\mu^{(\alpha, u_\alpha)} \big)= d_\alpha,
  \, \sum\limits_{v_\beta=1}^{\tilde{k}_\beta}\ell^*\big(\nu^{(\beta, v_\beta)} \big)= \tilde{d}_\beta }\nonumber\\
\hphantom{H^{({\bf d}, \tilde{\bf d})}_{({\bf q}, {\bf p})}(\mu, \nu)  :=}{}  \times
  \prod_{\alpha=1}^l W_{E(q_\alpha)}\big(\mu^{(\alpha,1)}, \dots , \mu^{(\alpha, k_\alpha)}\big)
\nonumber\\
\hphantom{H^{({\bf d}, \tilde{\bf d})}_{({\bf q}, {\bf p})}(\mu, \nu)  :=}{}\times
 \prod_{\beta=1}^m W_{H(p_\beta)}\big(\nu^{(\beta,1)}, \dots , \nu^{(\beta, \tilde{k}_\beta)}\big)
  H\big(\big\{\mu^{(\alpha, u_\alpha) };  \nu^{(\beta,  v_ \beta)}\big\}, \mu, \nu \big).
\end{gather}

Substituting the Frobenius--Schur formula (\ref{Frobenius_Schur_Hurwitz}) and the
Frobenius character formula~(\ref{Frobenius_character})
into (\ref{tau_qz_pw}), it follows that
\begin{gather}
\tau_{Q({\bf q}, {\bf z}; {\bf p},  {\bf w})}( {\bf t}, {\bf s}) := \tau_{Q({\bf q}, {\bf z}; {\bf p},  {\bf w})}(0,  {\bf t}, {\bf s})
\end{gather}
is the generating function for  $H^{({\bf c}, {\bf d})}_{({\bf q}, {\bf p})}(\mu, \nu)$.
Using multi-index notion to denote
\begin{gather}
\prod_{\alpha=1}^l w_\alpha^{d_\alpha} \prod_{\beta=1}^m z_\beta^{\tilde{d}_\beta} =:  {\bf w}^{\bf d} {\bf z}^{\tilde{\bf d}},
  \end{gather}
  we have
\begin{Theorem}
\begin{gather}
\tau^{Q({\bf q}, {\bf w}; {\bf p},  {\bf z})}( {\bf t}, {\bf s})  :=\sum_{{\bf d} \in \Nb}  {\bf w}^{\bf d} \sum_{\tilde{\bf d} \in \Nb }
{\bf z}^{\tilde{\bf d}}\sum_{\mu, \nu}
 H^{({\bf d}, \tilde{\bf d})}_{({\bf q}, {\bf p})}(\mu, \nu)p_\mu({\bf t}) p_\nu({\bf s}).
 \label{quantum_multispecies_geometric_tau}
 \end{gather}
\end{Theorem}

\subsection{Multispecies combinatorial quantum Hurwitz numbers}
\label{multispecies_quantum_combinatorial_hurwitz}

Another basis for $\Zb(\Cb[S_n])$  consists of the cycle sums
\begin{gather}
C_\mu := \sum_{h\in \cyc(\mu)} h,
\end{gather}
 where $\cyc(\mu)$ denotes the conjugacy class  of elements $h\in \cyc(\mu)$ with cycle lengths equal
to the parts of $\mu$.
 The two  are related by
\begin{gather}
F_\lambda = h_\lambda^{-1} \sum_{\mu, \, |\mu| = |\lambda|} \chi_\lambda(\mu) C_\mu,
\end{gather}
where $\chi_{\lambda}(\mu)$ denotes the irreducible character of the irreducible
representation of type~$\lambda$ eva\-luated on the conjugacy class~$\cyc(\mu)$.
Under the characteristic map, this   is equivalent to the Frobenius character formula~(\ref{Frobenius_character}).

Let $D_n$ denote the number of partitions of~$n$ and ${\bf F}^{(n, d_\alpha)}_{E(q_\alpha)}$, ${\bf F}^{(n, \tilde{d}_\beta)}_{H(p_\beta)}$   the $D_n \times D_n$ matrices whose  elements are $\big(F^{d_\alpha}_{E(q_\alpha)}(\mu, \nu)\big)_{|\mu|=|\nu|= n} $ and  $\big(F^{\tilde{d}_\beta}_{H(p_\beta)}(\mu, \nu)\big)_{|\mu|=|\nu|= n}$, respectively, for $\alpha=1, \dots, l$, $\beta=1, \dots, m$ as def\/ined in~(\ref{Fd_Eq}), (\ref{Fd_Hq}).  Since these  represent central elements of the group algebra~$\Cb[S_n]$, they commute amongst themselves. Def\/ining the matrix product
\begin{gather}
{\bf F}^{({\bf d}, \tilde{\bf d})}_{({\bf q}, {\bf p})} = \prod_{\alpha=1}^l {\bf F}^{(n, d_\alpha)}_{E(q_i)}   \prod_{\beta=1}^m {\bf F}^{(n, \tilde{d}_\beta)}_{H(p_\beta)} ,
\end{gather}
its matrix elements, denoted  $F^{({\bf c}, {\bf d})}_{({\bf q}, {\bf p})}  (\mu, \nu)$,  may be interpreted as the weighted
number of
\begin{gather}
d := \sum_{\alpha=1}^\alpha d_\alpha+ \sum_{\beta=1}^m \tilde{d}_\beta
\end{gather}
step paths in the Cayley graph from the conjugacy class of cycle  type $\mu$ to the one of type  $\nu$,
where~all paths are divided into equivalence classes, according to their {\em multisignatures} \linebreak
$\big\{\lambda^{(\alpha)}, \tilde{\lambda}^{(\beta)}\big\}_{\substack{\alpha =1, \dots, l \\ \beta =1, \dots , m}}$.

These consist  of a partition of the $d$  steps into $l +m$ parts, each of which is a subsequence assigned a
``colour'' and a ``class'' with $l$ of them of the f\/irst class and $m$ of  the second.
The number of partitions of f\/irst class with colour~$\alpha$ is~$d_\alpha$  while the number of
second class with colour~$\beta$ is~$\tilde{d}_\beta$. The partitions~$\lambda^{(\alpha)}$ of  weights~$d_\alpha$
are def\/ined to have parts~$\big\{ \lambda^{(\alpha)}_{u_\alpha}\big\}_{u_\alpha=1, \dots, k_\alpha}$ equal to the number of  a transposition appears within that subsequence, having the same second element, and similarly for
$\big\{\tilde{\lambda}^{(\beta)}_{v_\beta}\big\}_{v_\beta =1, \dots , \tilde{k}_\beta}$
with
\begin{gather}
k_\alpha= \ell\big( \lambda^{(\alpha)}\big), \qquad\tilde{k}_\beta= \ell\big( \tilde{\lambda}^{(\beta)}\big)
\end{gather}
the number of such parts.

For each such subpartition, the weight assigned is  the product of the weights for each subsegment
\begin{gather}
\prod_{\alpha=1}^l E_{\lambda^{(\alpha)}}(q_\alpha, w_\alpha)
\prod_{\beta=1}^m H_{\tilde{(\lambda)}^\beta}(p_\beta, z_\beta)
\end{gather}
and $F^{({\bf c}, {\bf d})}_{({\bf q}, {\bf p})}  (\mu, \nu)$  is the sum of these, each  multiplied by the number
of elements of the equivalence class of paths with the given multisignature
\begin{gather}
F^{({\bf c}, {\bf d})}_{({\bf q}, {\bf p})}  (\mu, \nu) = (n!)^{-l-m} \sum_{\mu^{(2)}, \dots ,  \mu^{(m+l)} } \prod_{\alpha=1}^l
 \left(\sum_{\{\lambda^{(\alpha)}\}} E_{\lambda^{(\alpha)}}(q_\alpha, w_\alpha) m^{\lambda^{(\alpha)}}_{\mu_\alpha \mu_{\alpha+1}} \right)
 \nonumber\\
 \hphantom{F^{({\bf c}, {\bf d})}_{({\bf q}, {\bf p})}  (\mu, \nu) =}{}\times
\prod_{\beta=1}^m \left(\sum_{\{\tilde{\lambda}^{(\beta)}\}}
H_{\tilde{\lambda}^{(\beta)}}(p_\beta, z_\beta) m^{\lambda^{(\beta)}}_{\mu_{\beta+l} , \mu_{\beta +l +1}} \right),
\end{gather}
where $(\mu, \nu):= \big(\mu^{(1)}, \mu^{(l+m+1)}\big)$.

The multispecies generalization of (\ref{Hd_equals_Fd}) is
equality of the geometric and combinatorial Hurwitz numbers:
\begin{Theorem}
\begin{gather}
F^{({\bf c}, {\bf d})}_{({\bf q}, {\bf p})}  (\mu, \nu) = H^{({\bf c}, {\bf d})}_{({\bf q}, {\bf p})}  (\mu, \nu).
\end{gather}
\end{Theorem}

\begin{proof}
Applying the central element~(\ref{central_qpzw_generator})
to the cycle sum $C_\mu$ gives
\begin{gather}
\hat{Q}({\bf q}, {\bf z}; {\bf p}, {\bf w}, \JJ)  C_\mu = \sum_{\nu, \, |\nu|=|\mu|}F^{({\bf c}, {\bf d})}_{({\bf q}, {\bf p})}  (\mu, \nu)   C_\nu
\end{gather}
and the orthogonality of group characters implies that
\begin{gather}
\tau^{Q({\bf q}, {\bf z}; {\bf p},  {\bf w})}({\bf t}, {\bf s})  :=\sum_{{\bf c}=(0, \dots,0);\, {\bf d} = (0, \dots, 0)}^ {(\infty, \dots, \infty)}
{\hskip -20 pt} {\bf z}^{\bf c} {\bf w}^{\bf d}\sum_{\mu, \nu}
 F^{({\bf c}, {\bf d})}_{({\bf q}, {\bf p})}(\mu, \nu)p_\mu({\bf t}) p_\nu({\bf s}).
 \end{gather}
 Comparing this with equation~(\ref{quantum_multispecies_geometric_tau}) proves the result.
\end{proof}

\begin{Remark}[multispecies Bose gases]
Returning to the interpretation of the quantum Hurwitz weighting distributions in terms
of Bose gases, if we choose each~$q_i$ to be a positive real number with $q_i<1$, and
parametrize it, as before,
\begin{gather}
q_i = e^{-\beta \hbar\omega_i}
\end{gather}
for some ground state energy $\hbar \omega_i$, and again choose a linear energy spectrum,
with energy assigned to the branchpoint $\mu$ of type $i$ with prof\/ile type $\mu$ to be
\begin{gather}
\epsilon^{(i)}(\mu) :=  \ell^*(\mu) \hbar \omega_i,
\end{gather}
we see that the resulting contributions to the weighting distributions distributions
 from each species of branch points of ramif\/ication type $\mu^{(j)}$ are given by
 \begin{gather}
n^{(i)}_{B}(\mu)=  {1\over e^{\beta \epsilon^{(i)}(\mu)}  -1},
 \end{gather}
 which are those of a multispecies  mixture of Bose gases.
\end{Remark}

\subsection*{Acknowledgments}

 This work is an extension of a~joint project~\cite{GH1, GH2} with  M.~Guay-Paquet,
in which the notion of inf\/inite parametric families of weighted Hurwitz numbers was f\/irst introduced,
combined with the notion of signed multispecies  Hurwitz numbers as introduced
in~\cite{HOr} with  A.Yu.~Orlov. The author would like to thank both these co-authors for helpful discussions.
Work supported by the Natural Sciences and Engineering Research Council of Canada (NSERC) and the Fonds de recherche du Qu\'ebec~-- Nature et technologies (FRQNT).

\pdfbookmark[1]{References}{ref}
\LastPageEnding

\end{document}